%% file: main.tex
\documentclass[conference]{IEEEtran}

\usepackage{amsfonts}
\usepackage{amssymb}
\usepackage{amsthm}
\usepackage{amsmath, amsfonts, amssymb}
\usepackage[dvips]{graphicx}
\usepackage{verbatim}
\usepackage{setspace}
\usepackage{bm}
\usepackage{algorithmic} 
\usepackage[ruled,vlined]{algorithm2e}
\usepackage{cite}

\usepackage{changepage}
\usepackage{pdfpages}
\usepackage{color}
\newtheorem{proposition}{Proposition}

\newcommand{\highlightcolor}{black}

\IEEEoverridecommandlockouts

\begin{document}
\title{\LARGE{Channel Autocorrelation Estimation for IRS-Aided Wireless Communications Based on Power Measurements}
\vspace{-13 pt}}

\author{\IEEEauthorblockN{Ge Yan$^\dag$, Lipeng Zhu$^\dag$, and Rui Zhang$^{\dag\ddag}$}
\IEEEauthorblockA{$^\dag$Department of Electrical and Computer Engineering, National University of Singapore, Singapore 117583.\\
$^{\ddag}$The Chinese University of Hong Kong, Shenzhen, and Shenzhen Research Institute of Big Data, Shenzhen, China 518172. \\
Email: geyan@u.nus.edu, zhulp@nus.edu.sg, elezhang@nus.edu.sg}
\vspace{-30 pt}}

\maketitle

\input{content/abstract.tex}
\input{content/intro.tex}

\input{content/system-model.tex}

\input{content/proposed.tex}

\input{content/simulation.tex}

\input{content/conclusion.tex}

%
\IEEEpeerreviewmaketitle


\vspace{-2pt}
\bibliographystyle{IEEEtran} 
\bibliography{IEEEabrv, reference}

\end{document}

%% file: content/abstract.tex
\begin{abstract}
    Intelligent reflecting surface (IRS) can bring significant performance enhancement for wireless communication systems by reconfiguring wireless channels via passive signal reflection. 
    However, such performance improvement generally relies on the knowledge of channel state information (CSI) for IRS-associated links. 
    Prior IRS channel estimation strategies mainly estimate IRS-cascaded channels based on the excessive pilot signals received at the users/base station (BS) with time-varying IRS reflections, which, however, are not compatible with the existing channel training/estimation protocol for cellular networks. 
    To address this issue, we propose in this paper a new channel estimation scheme for IRS-assisted communication systems based on the received signal power measured at the user, which is practically attainable without the need of changing the current protocol. 
    Specifically, due to the lack of signal phase information in power measurements, the autocorrelation matrix of the BS-IRS-user cascaded channel is estimated by solving equivalent matrix-rank-minimization problems. 
    Simulation results are provided to verify the effectiveness of the proposed channel estimation algorithm as well as the IRS passive reflection design based on the estimated channel autocorrelation matrix. 
\end{abstract}


%% file: content/intro.tex
\vspace{-16pt}
\section{INTRODUCTION}\label{sec:introduction}
    \IEEEPARstart{I}{ntelligent} reflecting surface (IRS) has received great attention in recent years due to its promising ability to reconfigure wireless channels. 
    By applying tunable phase shifts to incident wireless signals, IRS can effectively control their propagation channels and thereby significantly enhance the wireless communication performance such as transmission data rate and reliability~\cite{ref:IRSTutorial}. 
    Due to this benefit as well as its high deployment flexibility, low hardware cost and power consumption, IRS has been identified as a key enabling technology for the future sixth-generation (6G) wireless networks~\cite{ref:IRSTutorial}. 
    However, to reap the high performance gain by IRS, it is essential to acquire the channel state information (CSI) for the IRS channels with its assisting base station (BS) and users. 
    This is practically challenging because the passive IRS is not equipped with wireless transceivers, which results in that the BS-IRS and IRS-user channels cannot be separately estimated in general, while only their cascaded (i.e., BS-IRS-user) channel can be estimated~\cite{ref:IRS-Survey}. 
    However, to compensate for the significant
    product-distance path loss of the IRS-cascaded link, the number of IRS reflecting elements needs to be sufficiently large (e.g., hundreds or even thousands~\cite{ref:Physics-based-modeling-IRS}) in practice and their individual cascaded channels are generally different, which may incur high pilot signal overhead for channel estimation in IRS-assisted communication systems. 
    
    Most existing works on IRS-cascaded channel estimation focus on conventional pilot-based methods by exploiting the IRS time-varying reflection~\cite{ref:low-comp-ON-OFF, ref:IRS-OFDM-CE-DFT, ref:ce-bf-discrete}. 
    For example, by switching on only one reflecting element at one time, the authors in~\cite{ref:low-comp-ON-OFF} estimated the IRS-cascaded channel based on the received signals at each user. 
    To exploit the aperture gain of IRS for channel estimation, a discrete-Fourier-transform (DFT) based IRS reflection pattern was proposed in~\cite{ref:IRS-OFDM-CE-DFT} with all reflecting elements switched on for conducting the minimum-mean-square-error (MMSE) estimation of the cascaded channel in an IRS-assisted orthogonal-frequency-division-multiplexing (OFDM) system, while a Hadamard matrix based IRS reflection pattern was designed in~\cite{ref:ce-bf-discrete} under the practical discrete phase shift constraint on IRS reflection. 
    Moreover, IRS channel estimation was formulated as a compressed sensing problem in~\cite{ref:CS-IRS-mmW} by exploiting the sparsity of the channel paths in the angle domain to reduce the number of pilots required. 
    In addition, deep residual network was employed in~\cite{ref:deep-residual-ce} to refine the least-sqaure (LS) estimation of IRS channels. 
    
    In the aforementioned works, CSI is estimated based on the received complex-valued pilot signals at the users/BS. 
    However, in the protocol of existing wireless communication systems such as 4G/5G~\cite{ref:3gpp:38.211}, the pilots are dedicated to estimating the BS-user direct channels only. 
    As such, to estimate the new IRS-cascaded CSI as in the existing works, substantial additional pilots are required, which thus needs to significantly modify the existing channel estimation/training protocol. 
    To tackle this problem, IRS reflection designs based on the received signal power measurements at the users have been proposed, which do not require additional pilot signals for explicit IRS CSI estimation. 
    As user power measurements are commonly adopted and easy to obtain in existing wireless systems, such as reference signal received power (RSRP), this approach can be practically implemented without any change of the current protocol. 
    For example, the authors in~\cite{ref:RFocus} and~\cite{ref:CSM} proposed to design the IRS reflection coefficients based on received signal power measurements with different IRS reflections over time. 
    Specifically, based on power measurements at the user, each IRS reflecting element sets its reflection coefficient which achieves the maximum expectation of the received power conditioned on it, which is called the conditional sample mean (CSM) method in~\cite{ref:CSM}. 
    However, to obtain an accurate estimation of the conditional expectation for CSM, an excessively large amount of IRS training reflections/power measurements are generally required (in the quadratic order of the number of IRS reflecting elements), which is still time-consuming for practical implementation. 
    
    It is worth noting that the CSM-based methods in~\cite{ref:RFocus, ref:CSM} did not fully exploit the power measurements to obtain partial CSI of IRS-cascaded channels, thus resulting in their high overhead for power measurements and low IRS beamforming gain. 
    To improve the existing IRS channel estimation/beamforming design based on user power measurements, this paper proposes a new channel autocorrelation estimation scheme. 
    Specifically, the autocorrelation matrix of IRS-cascaded channel is estimated based on received signal power measurements at the user with randomly generated IRS reflections over time, and then the IRS reflection is designed based on the estimated channel autocorrelation matrix for data transmission. 
    In particular, the channel autocorrelation estimation problem is equivalently transformed into rank-minimization problems and alternating optimization is employed to obtain efficient solutions. 
    Simulation results validate the effectiveness of the proposed channel estimation scheme based on power measurements, and demonstrate the superior performance of IRS reflection design based on the estimated channel autocorrelation matrix compared to other benchmark schemes such as CSM.

%% file: content/system-model.tex
\section{System Model}\label{sec:system-model}

        \vspace{-10pt}
        \begin{figure}[ht]
            \begin{center}
                \includegraphics[scale = 0.13]{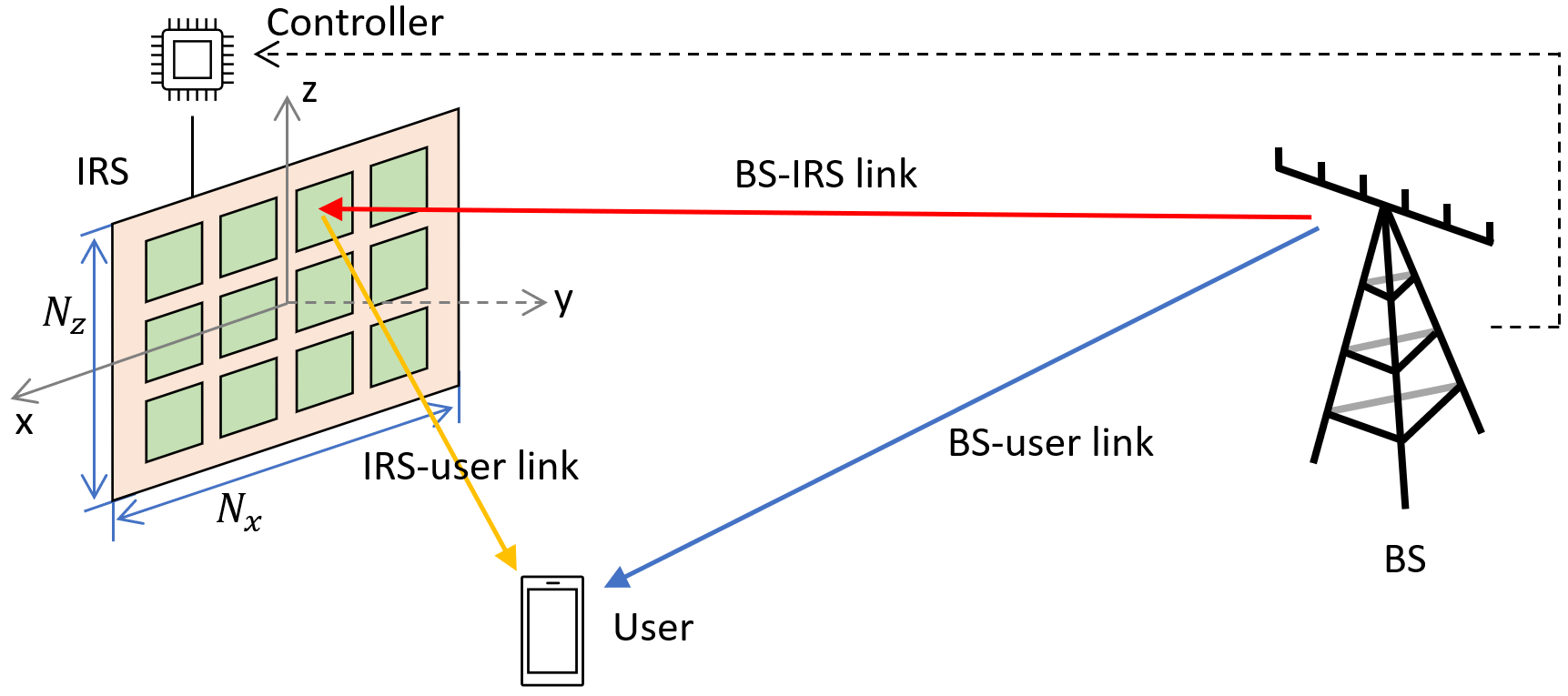}
                \vspace{-7pt}
                \caption{An IRS-aided wireless communication system. }
                \label{Fig:system}
            \end{center}
            \vspace{-15pt}
        \end{figure}
        
        As shown in Fig.~\ref{Fig:system}, we consider an IRS-aided downlink communication system with a single-antenna BS (or equivalently, multi-antenna BS with transmit precoding fixed) serving a single-antenna user, where an IRS is deployed to establish a reflected link to assist in their communication. 
        The IRS is composed of $N_{irs} = N_{x}N_{z}$ reflecting elements, where $N_{x}$ and $N_{z}$ are the number of reflecting elements in the horizontal and vertical dimensions, respectively. 
        Let $u_n$ denote the reflection coefficient of the $n$-th element, $n = 1, 2, \ldots, N_{irs}$, while $\boldsymbol{u} = [u_1, \ldots, u_{N_{irs}}]^T\in\mathbb{C}^{N_{irs}\times 1}$ and $\boldsymbol{\Theta} = \text{diag}(\boldsymbol{u}) \in\mathbb{C}^{N_{irs}\times N_{irs}}$ denote the IRS reflection coefficients vector and matrix, respectively. 
        Due to the unit amplitude constraint on the reflecting coefficients, we have $|u_n| = 1$ for $\forall n$. 
        Furthermore, denote the number of bits for controlling the discrete phase shift of each element as $b$. 
        Then, the reflection coefficient $u_n$ should be selected from a discrete set $\Phi_b = \{e^{j\Delta\theta}, \ldots, e^{j2^b\Delta\theta}\}$, with $\Delta\theta = 2\pi/2^b$. 
        Denoting $\Phi_b^M = \{\boldsymbol{x}\in\mathbb{C}^{M\times 1} | x_n\in\Phi_b, n = 1, \ldots, M\}$ as the set of $M$-dimensional vectors whose elements are selected from $\Phi_b$, we thus have $\boldsymbol{u}\in\Phi_b^{N_{irs}}$. 

        The baseband equivalent channels of the BS-IRS link, BS-user link, and IRS-user link are denoted as $\boldsymbol{g}\in\mathbb{C}^{N_{irs}\times 1}$, $h_d\in\mathbb{C}$, and $\boldsymbol{h}_r\in\mathbb{C}^{N_{irs}\times 1}$, respectively. 
        The received signal at the user is thus given by 
        \begin{equation}\label{def:received-sig}
            y = \left(\boldsymbol{h}_{r}^{H}\boldsymbol{\Theta}\boldsymbol{g} + h_{d}^*\right)s + z, 
        \end{equation}
        where $s$ is the transmitted signal with power $p_0$ and $z\sim\mathcal{CN}(0, \sigma^2)$ denotes the noise. 
        According to~\cite{ref:3gpp:36.214}, the RSRP is measured as the average power of multiple received reference signals, and thus the effect of noise can be mitigated to an arbitrarily low level. 
        Specifically, it is assumed in this paper that the RSRP is obtained by taking the average power over a sufficiently large number of received reference signals, which is given by
        \begin{equation}\label{def:received-power}
            r = \mathbb{E}\left[|y|^2\right] 
            = p_0\left|\left(
                \boldsymbol{h}_{r}^{H}\boldsymbol{\Theta}\boldsymbol{g} + h_{d}^*
            \right)\right|^2 + \sigma^2. 
        \end{equation}
        By deducting the noise power $\sigma^2$ from RSRP, the effect of noise can be removed and the noiseless signal power measurement $p = p_0|(\boldsymbol{h}_{r}^{H}\boldsymbol{\Theta}\boldsymbol{g} + h_{d}^*)|^2$ can be obtained. 
        Due to the relation $\boldsymbol{h}_{r}^{H}\boldsymbol{\Theta} = \boldsymbol{h}_{r}^{H}\text{diag}(\boldsymbol{u}) = \boldsymbol{u}^{H}\text{diag}(\boldsymbol{h}_{r}^{H})$, the power measurement can be rewritten as $p = p_0|\boldsymbol{u}^{H}\text{diag}(\boldsymbol{h}_{r}^{H})\boldsymbol{g} + h_{d}^*|^2$. 
        To further simplify the notation, let $N = N_{irs} + 1$ and define $\bar{\boldsymbol{h}} = \sqrt{p_0}[\boldsymbol{g}^H\text{diag}(\boldsymbol{h}_r), h_d]^H\in\mathbb{C}^{N\times 1}$ as the equivalent channel and $\boldsymbol{v} = [\boldsymbol{u}^T, 1]^T\in\Phi_b^{N}$ as the equivalent IRS reflection vector. 
        Then, the signal power measurement can be represented by $p = \left|\boldsymbol{v}^H\bar{\boldsymbol{h}}\right|^2 = \text{tr}(\bar{\boldsymbol{H}}\boldsymbol{V})$, 
        where $\bar{\boldsymbol{H}} = \bar{\boldsymbol{h}}\bar{\boldsymbol{h}}^H$ and $\boldsymbol{V} = \boldsymbol{v}\boldsymbol{v}^H$ are the autocorrelation matrices of the equivalent channel $\bar{\boldsymbol{h}}$ and the equivalent IRS reflection vector $\boldsymbol{v}$, respectively.

\vspace{-5pt}
\section{Channel Autocorrelation Matrix Estimation}\label{subsec:estimation-framework}
        We assume that the user is quasi-static and its channels with the BS/IRS do not change for a long time, during which the IRS changes its reflection coefficients $T$ times in total and in the meanwhile, the user measures the corresponding received signal power values and feed them back to the BS or some other processing unit that can design the IRS reflection coefficients and send them to the IRS controller for implementation via a separate wireless link. 
        Specifically, for the $t$-th power measurement, $t = 1, \ldots, T$, a random IRS reflection vector $\boldsymbol{u}_t$ is applied. 
        Define $\boldsymbol{v}_t = [\boldsymbol{u}_t, 1]^T$ and let $p_t = |\boldsymbol{v}_t^H\bar{\boldsymbol{h}}|^2 = \text{tr}(\bar{\boldsymbol{H}}\boldsymbol{V}_t)$ denote the signal power measured by the user, with $\boldsymbol{V}_t = \boldsymbol{v}_t\boldsymbol{v}_t^H$. 
        With all the power measurements obtained, $\bar{\boldsymbol{H}}$ is first estimated based on $p_t$ and $\boldsymbol{V}_t$, $t = 1, \ldots, T$. 
        Then, the IRS reflection vector is designed based on the estimated channel autocorrelation matrix to maximize the effective channel gain between the BS and user. 
        
        Given power measurements $\boldsymbol{p} = [p_1, \ldots, p_T]^T$ and $\boldsymbol{v}_t$, $t = 1, \ldots, T$, the channel autocorrelation matrix estimation problem can be formulated as finding a rank-one matrix $\boldsymbol{H}$ that satisfies $\text{tr}(\boldsymbol{H}\boldsymbol{V}_t) = p_t$, $t = 1, \ldots, T$, i.e., 
        \begin{subequations}\label{prob:cov-est-find-origin}
            \begin{align}
                & \text{find} \ {\boldsymbol{H}}\in\mathbb{S}_{+}^{N\times N} \tag{\ref{prob:cov-est-find-origin}} \\
                & ~ \mathrm{s.t.} \ \text{tr}(\boldsymbol{H}\boldsymbol{V}_t) = p_t, t = 1, \ldots, T, \label{prob:cov-est-find-power} \\
                & ~~~~~~ \text{rank}(\boldsymbol{H}) = 1, \label{prob:cov-est-find-rank-one}
            \end{align}
        \end{subequations}
        where $\mathbb{S}_{+}^{N\times N}$ denotes the set of all positive semidefinite hermitian matrices of dimension $N\times N$. 
        The form of Problem~\eqref{prob:cov-est-find-origin} is the same as the PhaseLift problem studied in~\cite{ref:PhaseLift}, where the trace-minimization relaxation was applied to find an approximate solution. 
        However, the performance of the approximate solution relies on the assumption that vectors $\{\boldsymbol{v}_t, t = 1, \ldots, T\}$ are independently and identically distributed (i.i.d.) Guassian random vectors, which are not applicable to IRS due to its unit-amplitude reflection with discrete phase shifts. 
        Thus, new methods are required to solve Problem~\eqref{prob:cov-est-find-origin} efficiently. 

        Next, we analyze the existence and uniqueness of the solution for Problem~\eqref{prob:cov-est-find-origin}. 
        Due to the practical discrete phase shifts for IRS, the solutions for Problem~\eqref{prob:cov-est-find-origin} are different for $b = 1$ and $b \ge 2$, which leads to the following proposition. 

        \begin{proposition}\label{prop:solution-set}
            For sufficiently large $N$ and $T$, Problem~\eqref{prob:cov-est-find-origin} has one unique solution $\bar{\boldsymbol{H}}$ for $b\ge 2$, while for $b = 1$, it has two solutions, i.e., $\bar{\boldsymbol{H}}$ and its conjugate matrix $\bar{\boldsymbol{H}}^*$. 
        \end{proposition}
        
        \begin{proof}[Proof (sketched)\textup{:}\nopunct]\label{proof:solution-set}
            The set of $(N\times N)$-dimensional hermitian matrices forms an $N^2$-dimensional linear space. 
            Thus, the power constraints in~\eqref{prob:cov-est-find-power}, which are linear to matrix $\boldsymbol{H}$, define an affine subspace in the hermitian matrix space. 

            For $b\ge 2$, $\boldsymbol{v}_t\in\Phi_b^{N}$, $t = 1, \ldots, T$, is a complex vector. 
            For sufficiently large $N$ and $T$, it can be proved that the affine subspace confined by~\eqref{prob:cov-est-find-power} is tangent to the manifold of rank-one hermitian matrices defined by~\eqref{prob:cov-est-find-rank-one} at one point, $\bar{\boldsymbol{H}}$, in the hermitian matrix space, which means that $\bar{\boldsymbol{H}}$ is the only rank-one hermitian matrix feasible to constraints~\eqref{prob:cov-est-find-power}. 

            For $b = 1$, $\boldsymbol{v}_t\in\Phi_1^{N} = \{\pm 1\}^N$, $t = 1, \ldots, T$, is always a real vector. 
            By applying eigenvalue decomposition, any positive semidefinite hermitian matrix $\boldsymbol{H}$ can be written as $\boldsymbol{H} = \boldsymbol{Q}\boldsymbol{\Lambda}\boldsymbol{Q}^H$, where $\boldsymbol{Q}\in\mathbb{C}^{N\times N}$ is a unitary matrix and $\boldsymbol{\Lambda} = \text{diag}(\alpha_1, \alpha_2, \ldots, \alpha_N)$ with $\alpha_1, \ldots, \alpha_N \ge 0$. 
            Define $\boldsymbol{\Sigma} = \boldsymbol{Q}\text{diag}(\sqrt{\alpha_1}, \sqrt{\alpha_2}, \ldots, \sqrt{\alpha_N})$, and $\boldsymbol{\Sigma}_{\text{re}} = \text{Re}(\boldsymbol{\Sigma})$ and $\boldsymbol{\Sigma}_{\text{im}} = \text{Im}(\boldsymbol{\Sigma})$ as the real and imaginary parts of $\boldsymbol{\Sigma}$. 
            Then, we have $\boldsymbol{H} = \boldsymbol{\Sigma}\boldsymbol{\Sigma}^H = (\boldsymbol{\Sigma}_{\text{re}}\boldsymbol{\Sigma}_{\text{re}}^T + \boldsymbol{\Sigma}_{\text{im}}\boldsymbol{\Sigma}_{\text{im}}^T) + j(\boldsymbol{\Sigma}_{\text{im}}\boldsymbol{\Sigma}_{\text{re}}^T - \boldsymbol{\Sigma}_{\text{re}}\boldsymbol{\Sigma}_{\text{im}}^T)$ and thus
            \begin{equation}\label{eq:binary-power}
                \begin{aligned}
                    \text{tr}(\boldsymbol{H}\boldsymbol{V}_t) & = \|\boldsymbol{\Sigma}^H\boldsymbol{v}_t\|_2^2 = \|\boldsymbol{\Sigma}_{\text{re}}^T\boldsymbol{v}_t - j\boldsymbol{\Sigma}_{\text{im}}^T\boldsymbol{v}_t\|_2^2 \\
                    & = \|\boldsymbol{\Sigma}_{\text{re}}^H\boldsymbol{v}_t\|_2^2 + \|\boldsymbol{\Sigma}_{\text{im}}^H\boldsymbol{v}_t\|_2^2 = \text{tr}(\boldsymbol{H}_r\boldsymbol{V}_t), 
                \end{aligned}
            \end{equation}
            where $\boldsymbol{H}_r = \text{Re}(\boldsymbol{H}) = \boldsymbol{\Sigma}_{\text{re}}\boldsymbol{\Sigma}_{\text{re}}^T + \boldsymbol{\Sigma}_{\text{im}}\boldsymbol{\Sigma}_{\text{im}}^T$ is the real part of $\boldsymbol{H}$. 
            This means that the imaginary part of $\boldsymbol{H}$ does not influence the received signal power at the user, i.e., $p_t$. 
            Thus, for any solution $\hat{\boldsymbol{H}}$ for Problem~\eqref{prob:cov-est-find-origin}, its conjugate $\hat{\boldsymbol{H}}^*$ is also a solution because it satisfies $\text{tr}(\hat{\boldsymbol{H}}^*\boldsymbol{V}_t) = \text{tr}(\hat{\boldsymbol{H}}_r\boldsymbol{V}_t) = \text{tr}(\hat{\boldsymbol{H}}\boldsymbol{V}_t) = p_t$ for $\forall t$ and is a rank-one matrix, with $\hat{\boldsymbol{H}}_r = \text{Re}(\hat{\boldsymbol{H}}) = \text{Re}(\hat{\boldsymbol{H}}^{*})$. 
            For sufficiently large $N$ and $T$, it can be proved that the affine subspace confined by~\eqref{prob:cov-est-find-power} is tangent to the manifold of rank-one hermitian matrices defined by~\eqref{prob:cov-est-find-power} at exactly two points, $\bar{\boldsymbol{H}}$ and $\bar{\boldsymbol{H}}^*$, in the hermitian matrix space. 
        \end{proof}

        As the solutions for $b = 1$ and $b\ge 2$ are different, we solve Problem~\eqref{prob:cov-est-find-origin} for these two cases separately. 
        For the case of $b\ge 2$, we derive the unique autocorrelation matrix $\bar{\boldsymbol{H}}$ by directly solving Problem~\eqref{prob:cov-est-find-origin}. 
        For the case of $b = 1$, the autocorrelation matrix cannot be uniquely determined, but $\bar{\boldsymbol{H}}_r = \text{Re}(\bar{\boldsymbol{H}}) = \text{Re}(\bar{\boldsymbol{H}}^*)$ is unique according to the proof of Proposition~\ref{prop:solution-set}. 
        Since the IRS reflection vector $\boldsymbol{u}$, as well as $\boldsymbol{v}$, is always a real vector for $b = 1$, the received signal power at the user can be written as $p = \text{tr}(\bar{\boldsymbol{H}}\boldsymbol{V}) = \text{tr}(\bar{\boldsymbol{H}}_r\boldsymbol{V})$. 
        Therefore, we only need to estimate $\bar{\boldsymbol{H}}_r$ for optimizing the IRS reflection vector for data transmission. 
        As such, for $b = 1$, we consider the following problem:
        \begin{subequations}\label{prob:cov-est-find-rank-two}
            \begin{align}
                & \text{find} \ {\boldsymbol{H}_r}\in\mathbb{S}_{+}^{N\times N}\cap\mathbb{R}^{N\times N} \tag{\ref{prob:cov-est-find-rank-two}} \\
                & ~ \mathrm{s.t.} \ \text{tr}(\boldsymbol{H}_r\boldsymbol{V}_t) = p_t, \ t = 1, \ldots, T, \label{prob:cov-est-find-rank-two-power} \\
                & ~~~~~~ \text{rank}(\boldsymbol{H}_r) \le 2. \label{prob:cov-est-find-rank-two-rank}
            \end{align}
        \end{subequations}
        The following proposition ensures the existence and uniqueness of the solution for Problem~\eqref{prob:cov-est-find-rank-two}. 
        \begin{proposition}\label{prop:solution-set-cov-real}
            For sufficiently large $N$ and $T$, Problem~\eqref{prob:cov-est-find-rank-two} has one unique solution $\bar{\boldsymbol{H}}_r$. 
        \end{proposition}
        \begin{proof}[Proof\textup{:}\nopunct]
            Obviously, $\bar{\boldsymbol{H}}_r$ is a solution for Problem~\eqref{prob:cov-est-find-rank-two}, which guarantees the existence of the solution. 
            For the uniqueness, any solution for Problem~\eqref{prob:cov-est-find-rank-two}, denoted by $\hat{\boldsymbol{H}}_r$, is symmetric and semidefinite with $\text{rank}(\hat{\boldsymbol{H}}_r) \le 2$. 
            By leveraging eigenvalue decomposition, we have $\hat{\boldsymbol{H}}_r = \boldsymbol{Q}_r\text{diag}(\alpha_1, \alpha_2)\boldsymbol{Q}_r^T$, with $\boldsymbol{Q}_r = [\boldsymbol{q}_{r1}, \boldsymbol{q}_{r2}]\in\mathbb{R}^{N\times 2}$ and $\alpha_1, \alpha_2 \ge 0$. 
            Denote $\hat{\boldsymbol{h}}_1 = \sqrt{\alpha_1}\boldsymbol{q}_{r1}, \hat{\boldsymbol{h}}_2 = \sqrt{\alpha_2}\boldsymbol{q}_{r2}\in\mathbb{R}^{N\times 1}$, then $\hat{\boldsymbol{H}}_r = \hat{\boldsymbol{h}}_1\hat{\boldsymbol{h}}_1^T + \hat{\boldsymbol{h}}_2\hat{\boldsymbol{h}}_2^T$. 
            Consider $\hat{\boldsymbol{h}} = \hat{\boldsymbol{h}}_1 + j\hat{\boldsymbol{h}}_2$ and $\hat{\boldsymbol{H}} = \hat{\boldsymbol{h}}\hat{\boldsymbol{h}}^H$. 
            Obviously, we have $\hat{\boldsymbol{H}}_r = \text{Re}(\hat{\boldsymbol{H}})$, $\text{rank}(\hat{\boldsymbol{H}}) = 1$ and $\text{tr}(\hat{\boldsymbol{H}}\boldsymbol{V}_t) = \text{tr}(\hat{\boldsymbol{H}}_r\boldsymbol{V}_t) = p_t$, $t = 1, \ldots, T$. 
            Thus, $\hat{\boldsymbol{H}}$ is a solution for Problem~\eqref{prob:cov-est-find-origin}, which means $\hat{\boldsymbol{H}} = \bar{\boldsymbol{H}}$ or $\bar{\boldsymbol{H}}^{*}$ according to Proposition~\ref{prop:solution-set}. 
            Then, we have $\hat{\boldsymbol{H}}_r = \text{Re}(\hat{\boldsymbol{H}}) = \bar{\boldsymbol{H}}_r$, which is uniquely determined. 
        \end{proof}

%% file: content/proposed.tex
\vspace{-3pt}
\section{Proposed Solutions}\label{sec:proposed-RX}
    In this section, the channel autocorrelation matrix estimation problems for $b\ge 2$ and $b = 1$, i.e., Problems~\eqref{prob:cov-est-find-origin} and~\eqref{prob:cov-est-find-rank-two}, are respectively solved by transforming them into equivalent matrix-rank-minimization problems.

\input{content/proposed-ratio-max.tex}

%% file: content/proposed-ratio-max.tex
\subsection{Solution for $b\ge 2$}\label{subsec:ratio-max-est-multiary}
    For $b\ge 2$, instead of solving Problem~\eqref{prob:cov-est-find-origin} directly, we consider the following rank-minimization problem: 
    \begin{subequations}\label{prob:cov-est-rank-min}
        \allowdisplaybreaks
        \begin{align}
            & \mathop{\min_{\boldsymbol{H}}} \ \text{rank}(\boldsymbol{H}) \tag{\ref{prob:cov-est-rank-min}} \\
            & ~ \mathrm{s.t.} \ \text{tr}(\boldsymbol{H}\boldsymbol{V}_t) = p_t, \ t = 1, \ldots, T, \label{prob:cov-est-rank-min-power} \\
            & ~~~~~~ {\boldsymbol{H}}\in\mathbb{S}_{+}^{N\times N}. \label{prob:cov-est-rank-min-semidefinite}
        \end{align}
    \end{subequations}
    The equivalence between Problems~\eqref{prob:cov-est-find-origin} and~\eqref{prob:cov-est-rank-min} is analyzed as follows. 
    As the channel autocorrelation matrix $\bar{\boldsymbol{H}}$ is feasible to Problem~\eqref{prob:cov-est-rank-min}, any optimal solution for this problem, denoted by $\hat{\boldsymbol{H}}$, should satisfy $0 < \text{rank}(\hat{\boldsymbol{H}}) \le \text{rank}(\bar{\boldsymbol{H}}) = 1$, which leads to $\text{rank}(\hat{\boldsymbol{H}}) = 1$ and thus $\hat{\boldsymbol{H}}$ is also a solution for Problem~\eqref{prob:cov-est-find-origin}. 
    Reversely, any solution $\hat{\boldsymbol{H}}'$ for Problem~\eqref{prob:cov-est-find-origin} is feasible to the rank-minimization problem~\eqref{prob:cov-est-rank-min} and satisfies $\text{rank}(\hat{\boldsymbol{H}}') = 1$, which indicates that $\hat{\boldsymbol{H}}'$ is an optimal solution for Problem~\eqref{prob:cov-est-rank-min}. 
    Thus, matrix $\bar{\boldsymbol{H}}$ can be equivalently estimated by solving the rank-minimization problem~\eqref{prob:cov-est-rank-min}. 
    
    Since $\boldsymbol{H}\in\mathbb{S}_{+}^{N\times N}$ is nonzero, all the eigenvalues of $\boldsymbol{H}$ are real and non-negative, and $\text{tr}(\boldsymbol{H}) > 0$. 
    Define the eigenvalue-ratio function for matrix $\boldsymbol{H}$ as
    \begin{equation}\label{def:lambda-ratio-func}
        g(\boldsymbol{H}) = \frac{\lambda_{\text{max}}(\boldsymbol{H})}{\text{tr}(\boldsymbol{H})}, 
    \end{equation}
    where $\lambda_{\text{max}}(\boldsymbol{H})$ is the largest eigenvalue of $\boldsymbol{H}$. 
    Obviously, $0 < g(\boldsymbol{H})\le 1$ for any nonzero $\boldsymbol{H}\in\mathbb{S}_{+}^{N\times N}$, and it is easy to verify that $\text{rank}(\boldsymbol{H}) = 1$ if and only if $g(\boldsymbol{H}) = 1$. 
    As we have mentioned above, any solution $\hat{\boldsymbol{H}}$ for Problem~\eqref{prob:cov-est-rank-min} satisfies $\text{rank}(\hat{\boldsymbol{H}}) = 1$. 
    Thus, we have $g(\hat{\boldsymbol{H}}) = 1$, which means that $\hat{\boldsymbol{H}}$ maximizes the eigenvalue-ratio function $g(\boldsymbol{H})$. 
    On the other hand, any matrix $\hat{\boldsymbol{H}}'$ that maximizes $g(\boldsymbol{H})$ subject to constraints~\eqref{prob:cov-est-rank-min-power} and~\eqref{prob:cov-est-rank-min-semidefinite} also minimizes $\text{rank}(\boldsymbol{H})$. 
    Therefore, the solutions for Problem~\eqref{prob:cov-est-rank-min} are the same as the solutions that maximize the eigenvalue-ratio function $g(\boldsymbol{H})$ subject to constraints~\eqref{prob:cov-est-rank-min-power} and~\eqref{prob:cov-est-rank-min-semidefinite}. 
    Note that $\lambda_{\text{max}}(\boldsymbol{H}) = \mathop{\max_{\|\boldsymbol{x}\|_2\le 1}}{\boldsymbol{x}^H\boldsymbol{H}\boldsymbol{x}}$. 
    Then, Problem~\eqref{prob:cov-est-rank-min} can be rewritten as 
    \begin{equation}\label{prob:cov-est-maxmax}
        \mathop{\max_{\boldsymbol{H}}\max_{\|\boldsymbol{x}\|_2 \le 1}} \ f(\boldsymbol{H}, \boldsymbol{x}) = \frac{\boldsymbol{x}^H\boldsymbol{H}\boldsymbol{x}}{\text{tr}(\boldsymbol{H})}, 
        ~~ \mathrm{s.t.} \ \eqref{prob:cov-est-rank-min-power},\eqref{prob:cov-est-rank-min-semidefinite}, 
    \end{equation}
    This optimization problem is non-convex, while alternating optimization can be employed to obtain a suboptimal solution for it. 
    Given $\boldsymbol{H}$, vector $\boldsymbol{x}$ can be easily optimized by applying eigenvalue decomposition for $\boldsymbol{H}$. 
    An optimal $\boldsymbol{x}$ is the normalized eigenvector of $\boldsymbol{H}$ corresponding to the largest eigenvalue. 
    Given $\boldsymbol{x}$, the optimization of $\boldsymbol{H}$ is simplified as 
    \begin{equation}\label{prob:cov-est-maxmax-cov}
            \mathop{\max_{\boldsymbol{H}}} \ \frac{\text{tr}(\boldsymbol{H}\boldsymbol{X})}{\text{tr}(\boldsymbol{H})}, 
            ~~ \mathrm{s.t.} \ \eqref{prob:cov-est-rank-min-power},\eqref{prob:cov-est-rank-min-semidefinite},
    \end{equation}
    with $\boldsymbol{X} = \boldsymbol{x}\boldsymbol{x}^H$. 
    This is a fractional programming problem and can be transformed into a convex optimization problem. 
    Specifically, define $\boldsymbol{G} = \boldsymbol{H}/\text{tr}(\boldsymbol{H})$ and $\gamma = 1/\text{tr}(\boldsymbol{H})$. 
    Then, Problem~\eqref{prob:cov-est-maxmax-cov} can be transformed into
    \begin{subequations}\label{prob:cov-est-maxmax-frac-progm}
        \allowdisplaybreaks
        \begin{align}
            & \mathop{\max_{\boldsymbol{G}, \gamma > 0}} \ \text{tr}(\boldsymbol{G}\boldsymbol{X}) \tag{\ref{prob:cov-est-maxmax-frac-progm}} \\
            & ~ \mathrm{s.t.} \ \text{tr}(\boldsymbol{G}\boldsymbol{V}_t) = p_t\gamma, \ t = 1, \ldots, T, \label{prob:cov-est-maxmax-frac-progm-power} \\
            & ~~~~~~ {\boldsymbol{G}}\in\mathbb{S}_{+}^{N\times N}, \ \text{tr}(\boldsymbol{G}) = 1, \label{prob:cov-est-maxmax-frac-progm-semidefinite}
        \end{align}
    \end{subequations}
    which is convex and thus can be solved by CVX~\cite{ref:Boyd-cvx}. 
    Denoting the optimal solution for Problem~\eqref{prob:cov-est-maxmax-frac-progm} as $\hat{\boldsymbol{G}}$ and $\hat{\gamma}$, then the optimal solution for Problem~\eqref{prob:cov-est-maxmax-cov} is obtained as $\hat{\boldsymbol{H}} = \hat{\boldsymbol{G}}/\hat{\gamma}$. 
    By solving $\boldsymbol{x}$ and $\boldsymbol{H}$ alternatively, an approximately rank-one matrix can be obtained when $g(\boldsymbol{H})$ exceeds a predefined ratio threshold $\epsilon\in(0, 1)$. 

    \vspace{-6pt}
    \begin{algorithm}[h]
        \caption{Proposed solution for Problem~\eqref{prob:cov-est-find-origin}. }
        \begin{algorithmic}[1]\label{alg:RX-est-multiary}
            \REQUIRE~$\{\boldsymbol{v}_t$, $t = 1, \ldots, T\}$, $\boldsymbol{p}\in\mathbb{R}^{T\times 1}$, ratio threshold $\epsilon$. 
            \STATE~Initialization: Solve $\boldsymbol{H}^{(0)}$ via trace-minimization relaxation~\cite{ref:PhaseLift}; iteration index $i\gets 1$. 
            \WHILE{$g(\boldsymbol{H}^{(i - 1)}) \le \epsilon$}
                \STATE~Let $\boldsymbol{x}^{(i)}$ be the normalized eigenvector of $\boldsymbol{H}^{(i -1)}$ corresponding to the largest eigenvalue, and $\boldsymbol{X}^{(i)} \gets \boldsymbol{x}^{(i)}{\boldsymbol{x}^{(i)}}^H$. 
                \STATE~Solve $\boldsymbol{G}^{(i)}$ and $\gamma^{(i)}$ from problem~\eqref{prob:cov-est-maxmax-frac-progm} with $\boldsymbol{X} = \boldsymbol{X}^{(i)}$, and obtain $\boldsymbol{H}^{(i)} \gets \boldsymbol{G}^{(i)}/\gamma^{(i)}$. 
                \STATE~$i\gets i + 1$. 
            \ENDWHILE
            \RETURN~The estimated matrix $\hat{\boldsymbol{H}}\gets\boldsymbol{H}^{(i - 1)}$. 
        \end{algorithmic}
    \end{algorithm}
    \vspace{-6pt}

    The proposed solution for Problem~\eqref{prob:cov-est-find-origin} for the case of $b\ge 2$ is summarized in Algorithm~\ref{alg:RX-est-multiary}. 
    The convergence of the proposed algorithm is analyzed as follows. 
    In Algorithm~\ref{alg:RX-est-multiary}, variables $\boldsymbol{H}$ and $\boldsymbol{x}$ are updated as $\boldsymbol{H}^{(i)}$ and $\boldsymbol{x}^{(i)}$ in the $i$-th iteration. 
    It is worth noticing that $\boldsymbol{x}^{(i)}$ maximizes $f(\boldsymbol{H}, \boldsymbol{x})$ given $\boldsymbol{H} = \boldsymbol{H}^{(i - 1)}$ and thus 
    \begin{equation}
        g(\boldsymbol{H}^{(i - 1)}) = \max_{\|\boldsymbol{x}\|_2 \le 1} f(\boldsymbol{H}^{(i - 1)}, \boldsymbol{x}) = f(\boldsymbol{H}^{(i - 1)}, \boldsymbol{x}^{(i)})
    \end{equation}
    holds for $i = 1, \ldots, I$. 
    Moreover, $\boldsymbol{H}^{(i)}$ is the optimal solution for Problem~\eqref{prob:cov-est-maxmax-cov} with $\boldsymbol{X} = \boldsymbol{X}^{(i)}$, which means 
    \begin{equation}
        \begin{aligned}
            f(\boldsymbol{H}^{(i)}, \boldsymbol{x}^{(i)}) & = \frac{\text{tr}(\boldsymbol{H}^{(i)}\boldsymbol{X}^{(i)})}{\text{tr}(\boldsymbol{H}^{(i)})} \ge \frac{\text{tr}(\boldsymbol{H}^{(i - 1)}\boldsymbol{X}^{(i)})}{\text{tr}(\boldsymbol{H}^{(i - 1)})} \\
            & = f(\boldsymbol{H}^{(i - 1)}, \boldsymbol{x}^{(i)}), \ i = 1, \ldots, I, 
        \end{aligned}
    \end{equation}
    where $I$ denotes the number of iterations. 
    Thus, we have 
    \begin{subequations}\label{eq:ratio-func-monotonic}
        \begin{align}
            g(\boldsymbol{H}^{(i)}) & = \max_{\|\boldsymbol{x}\|_2 \le 1} f(\boldsymbol{H}^{(i)}, \boldsymbol{x}) \ge f(\boldsymbol{H}^{(i)}, \boldsymbol{x}^{(i)}) \\
            & \ge f(\boldsymbol{H}^{(i - 1)}, \boldsymbol{x}^{(i)}) \\
            & = g(\boldsymbol{H}^{(i - 1)}), \ i = 1, \ldots, I. 
        \end{align}
    \end{subequations}
    Therefore, the eigenvalue-ratio function $g(\boldsymbol{H}^{(i)})$ is non-decreasing during the iterations for $i = 0, 1, \ldots, I$. 
    Since $g(\boldsymbol{H})$ is upper-bounded by $1$, the convergence of Algorithm~\ref{alg:RX-est-multiary} is guaranteed. 
    Meanwhile, the monotonic increase of the eigenvalue-ratio function indicates that matrix $\boldsymbol{H}^{(i)}$ gradually approaches a rank-one matrix over the iterations.

\subsection{Solution for $b = 1$}\label{subsec:ratio-max-est-binary}
    Next, we consider the case of $b = 1$. 
    To solve matrix $\bar{\boldsymbol{H}}_r$ from Problem~\eqref{prob:cov-est-find-rank-two}, a similar method to the case of $b\ge 2$ can be applied. 
    Consider the following rank-minimization problem: 
    \begin{subequations}\label{prob:cov-real-est-rank-min}
        \allowdisplaybreaks
        \begin{align}
            & \mathop{\min_{\boldsymbol{H}_r}} \ \text{rank}(\boldsymbol{H}_r) \tag{\ref{prob:cov-real-est-rank-min}} \\
            & ~ \mathrm{s.t.} \ \text{tr}(\boldsymbol{H}_r\boldsymbol{V}_t) = p_t, \ t = 1, \ldots, T, \label{prob:cov-real-est-rank-min-power} \\
            & ~~~~~~ {\boldsymbol{H}_r}\in\mathbb{S}_{+}^{N\times N}\cap\mathbb{R}^{N\times N}. \label{prob:cov-real-est-rank-min-semidefinite}
        \end{align}
    \end{subequations}
    The equivalence between Problems~\eqref{prob:cov-est-find-rank-two} and~\eqref{prob:cov-real-est-rank-min} is analyzed as follows. 
    Obviously, matrix $\bar{\boldsymbol{H}}_r = \text{Re}(\bar{\boldsymbol{H}})$ is feasible to Problem~\eqref{prob:cov-real-est-rank-min}. 
    Any optimal solution for Problem~\eqref{prob:cov-real-est-rank-min}, denoted by $\hat{\boldsymbol{H}}_r$, satisfies $\text{rank}(\hat{\boldsymbol{H}}_r)\le\text{rank}(\bar{\boldsymbol{H}}_r)\le 2$, which means that $\hat{\boldsymbol{H}}_r$ is also a solution for Problem~\eqref{prob:cov-est-find-rank-two}. 
    Thus, the solutions for Problem~\eqref{prob:cov-est-find-rank-two} can be obtained by solving the rank-minimization problem~\eqref{prob:cov-real-est-rank-min} equivalently. 

    Define the generalized eigenvalue-ratio function as
    \begin{equation}
        g_r(\boldsymbol{H}_r) = \frac{\lambda_1(\boldsymbol{H}_r) + \lambda_2(\boldsymbol{H}_r)}{\text{tr}(\boldsymbol{H}_r)}, 
    \end{equation}
    where $\lambda_1(\boldsymbol{H}_r)$ and $\lambda_2(\boldsymbol{H}_r)$ are the first and second largest eigenvalues of $\boldsymbol{H}_r$, respectively. 
    For nonzero $\boldsymbol{H}_r\in\mathbb{S}_{+}^{N\times N}$, we have $\text{tr}(\boldsymbol{H}_r) > 0$, $0 \le g_r(\boldsymbol{H}_r)\le 1$, and $g_r(\boldsymbol{H}_r) = 1$ holds if and only if $\text{rank}(\boldsymbol{H}_r)\le 2$. 
    Therefore, solving Problem~\eqref{prob:cov-real-est-rank-min} is equivalent to maximizing $g_r(\boldsymbol{H}_r)$ subject to constraints~\eqref{prob:cov-real-est-rank-min-power} and~\eqref{prob:cov-real-est-rank-min-semidefinite}. 
    Furthermore, since we have
    \begin{subequations}\label{prob:largest-two-eigenvalues}
        \begin{align}
            \lambda_1(\boldsymbol{H}_r) + \lambda_2(\boldsymbol{H}_r) = & \mathop{\max_{\boldsymbol{x}_1, \boldsymbol{x}_2}} \ \boldsymbol{x}_1^T\boldsymbol{H}_r\boldsymbol{x}_1 + \boldsymbol{x}_2^T\boldsymbol{H}_r\boldsymbol{x}_2 \tag{\ref{prob:largest-two-eigenvalues}} \\
            & ~ \mathrm{s.t.} \ \|\boldsymbol{x}_1\|_2 \le 1, \ \|\boldsymbol{x}_2\|_2\le 1, \label{prob:largest-two-eigenvalues-normed} \\
            & ~~~~~~ \boldsymbol{x}_1^T\boldsymbol{x}_2 = 0, \label{prob:largest-two-eigenvalues-orthogonal}
        \end{align}
    \end{subequations}
    Problem~\eqref{prob:cov-real-est-rank-min} can be equivalently written as 
    \begin{subequations}\label{prob:cov-est-maxmax-real}
        \begin{align}
            & \mathop{\max_{\boldsymbol{H}_r}\max_{\boldsymbol{x}_1, \boldsymbol{x}_2}} \ f_r(\boldsymbol{H}_r, \boldsymbol{x}_1, \boldsymbol{x}_2) = \frac{\boldsymbol{x}_1^T\boldsymbol{H}_r\boldsymbol{x}_1 + \boldsymbol{x}_2^T\boldsymbol{H}_r\boldsymbol{x}_2}{\text{tr}(\boldsymbol{H}_r)} \tag{\ref{prob:cov-est-maxmax-real}} \\
            & ~ \mathrm{s.t.} \ \eqref{prob:largest-two-eigenvalues-normed},~\eqref{prob:largest-two-eigenvalues-orthogonal},~\eqref{prob:cov-real-est-rank-min-power},~\eqref{prob:cov-real-est-rank-min-semidefinite}, \nonumber
        \end{align}
    \end{subequations}
    Alternating optimization can be applied to solve Problem~\eqref{prob:cov-est-maxmax-real} sub-optimally. 
    Given $\boldsymbol{H}_r$, the optimal $\boldsymbol{x}_1$ and $\boldsymbol{x}_2$ can be obtained as the normalized eigenvectors corresponding to the first and second largest eigenvalues of $\boldsymbol{H}_r$, respectively. 
    Given $\boldsymbol{x}_1$ and $\boldsymbol{x}_2$, $\boldsymbol{H}_r$ can be optimized via
    \begin{equation}\label{prob:cov-est-maxmax-real-cov}
            \mathop{\max_{\boldsymbol{H}_r}} \ \frac{\text{tr}(\boldsymbol{H}_r\boldsymbol{X})}{\text{tr}(\boldsymbol{H}_r)}, 
            ~~ \mathrm{s.t.} \ \eqref{prob:cov-real-est-rank-min-power},\eqref{prob:cov-real-est-rank-min-semidefinite}, 
    \end{equation}
    with $\boldsymbol{X} = \boldsymbol{x}_1\boldsymbol{x}_1^T + \boldsymbol{x}_2\boldsymbol{x}_2^T$. 
    This problem is also a fractional programming problem and can be solved in the same way as Problem~\eqref{prob:cov-est-maxmax-cov}. 
    The proposed solution for Problem~\eqref{prob:cov-est-find-rank-two} for the case of $b = 1$ is summarized in Algorithm~\ref{alg:RX-est-binary}, and its convergence can be guaranteed by the monotonic increase of the generalized eigenvalue-ratio function $g_r(\boldsymbol{H}_r^{(i)})$, similar to Algorithm~\ref{alg:RX-est-multiary}. 
    Due to the semidefinite poragmming applied for the optimization of $\boldsymbol{H}$ and $\boldsymbol{H}_r$, both Algorithms~\ref{alg:RX-est-multiary} and~\ref{alg:RX-est-binary} have a computational complexity of $\mathcal{O}(N^{4.5}I)$. 

    \vspace{-6pt}
    \begin{algorithm}[h]
        \caption{Proposed solution for Problem~\eqref{prob:cov-est-find-rank-two}. }
        \begin{algorithmic}[1]\label{alg:RX-est-binary}
            \REQUIRE~$\{\boldsymbol{v}_t$, $t = 1, \ldots, T\}$, $\boldsymbol{p}\in\mathbb{R}^{T\times 1}$, ratio threshold $\epsilon$. 
            \STATE~Initialization: Solve $\boldsymbol{H}_r^{(0)}$ via trace-minimization relaxation~\cite{ref:PhaseLift}; iteration index $i\gets 1$. 
            \WHILE{$g_r(\boldsymbol{H}_r^{(i - 1)}) \le \epsilon$}
                \STATE~Let $\boldsymbol{x}_1^{(i)}$ and $\boldsymbol{x}_2^{(i)}$ be the normlized eigenvectors of $\boldsymbol{H}_r^{(i - 1)}$ corresponding to the first and second largest eigenvalues, and $\boldsymbol{X}^{(i)} \gets \boldsymbol{x}_1^{(i)}{\boldsymbol{x}_1^{(i)}}^T + \boldsymbol{x}_2^{(i)}{\boldsymbol{x}_2^{(i)}}^T$. 
                \STATE~Solve $\boldsymbol{H}_r^{(i)}$ from problem~\eqref{prob:cov-est-maxmax-real-cov} with $\boldsymbol{X} = \boldsymbol{X}^{(i)}$. 
                \STATE~$i\gets i + 1$. 
            \ENDWHILE
            \RETURN~The estimated matrix $\hat{\boldsymbol{H}_r}\gets\boldsymbol{H}_r^{(i - 1)}$. 
        \end{algorithmic}
    \end{algorithm}
    \vspace{-6pt}

%% file: content/simulation.tex
\section{Performance Evaluation}\label{sec:performance-evaluation}
    \subsection{Simulation Setup}\label{subsec:setup}
        In the simulation, the BS and the IRS are located at $(50, -200, 20)$ and $(-2, -1, 0)$ in meters (m) in a three-dimensional coordinate system, respectively. 
        The location of the user is randomly generated as $(x_{u}, y_{u}, 0)$ with $0\le x_u, y_u\le 10$. 
        The size of IRS is set as $N_x\times N_z = 8\times 8 = 64$, and thus we have $N = 64 + 1 = 65$. 
        The path loss coefficient for all channels is given by $\eta = C_{0}d^{-\alpha}$, where $d$ is the signal propagation distance. 
        Additionally, $C_{0}$ and $\alpha$ are the channel gain at the reference distance of $1$ m and the path-loss exponent, which are denoted for the BS-user, BS-IRS, and IRS-user links as $C_{0, BU}$ and $\alpha_{BU}$, $C_{0, BI}$ and $\alpha_{BI}$, and $C_{0, IU}$ and $\alpha_{IU}$, respectively. 
        The corresponding path loss coefficients are denoted as $\eta_{BU}$, $\eta_{BI}$ and $\eta_{IU}$, respectively. 
        For all channels, Rician fading is assumed with the Rician factor denoted as $\beta_{BU}$, $\beta_{BI}$ and $\beta_{IU}$ for the BS-user, BS-IRS and IRS-user links, respectively. 
        Specifically, the expression of the BS-IRS channel vector $\boldsymbol{g}$ is given below as an example:
        \begin{equation}\label{def:rician-channel}
            \boldsymbol{g} = \sqrt{\frac{\beta_{BI}}{1 + \beta_{BI}}}\boldsymbol{g}^{\text{LoS}} + \sqrt{\frac{1}{1 + \beta_{BI}}}\boldsymbol{g}^{\text{NLoS}}. 
        \end{equation}
        Vector $\boldsymbol{g}^{\text{NLoS}} = \sqrt{\eta_{BI}}\tilde{\boldsymbol{g}}$ is the Gaussian non-line-of-sight (NLoS) component with $\tilde{\boldsymbol{g}}\sim\mathcal{CN}(\boldsymbol{0}, \boldsymbol{I})$ and $\boldsymbol{g}^{\text{LoS}}$ is the deterministic line-of-sight (LoS) component given by 
        \begin{equation}\label{def:mmW-channel}
                \boldsymbol{g}^{\text{LoS}} = \sqrt{\eta_{BI}}\boldsymbol{b}_{N}(\omega, \psi). 
        \end{equation}
        Vector $\boldsymbol{b}_{N}(\omega, \psi)\in\mathbb{C}^{N\times 1}$ is the steering vector of the LoS path from the BS to the IRS, 
        where $\omega\in[0, \pi)$ and $\psi\in[0, \pi)$ are the physical azimuth and elevation angles of arrival (AoAs) at the IRS, respectively. 
        Specifically, define $\boldsymbol{a}_{N}(\phi) = [e^{j\pi 0\phi}, \ldots, e^{j\pi (N - 1)\phi}]^T\in\mathbb{C}^{N\times 1}$ as the $N$-dimensional steering vector. Then, $\boldsymbol{b}_{N}(\omega, \psi)$ is defined as $\boldsymbol{b}_{N}(\omega, \psi) = \boldsymbol{a}_{N_x}(\cos{(\omega)}\sin{(\psi)})\otimes\boldsymbol{a}_{N_z}(\cos{(\psi)})$, where $\otimes$ is the Kronecker product. 
        In addition, we set $C_{0, BU} = -33$ dB, $C_{0, BI} = C_{0, IU} = -30$ dB, $\alpha_{BU} = 3.7$, $\alpha_{BI} = \alpha_{IU} = 2$, $\beta_{BU} = 0$, $\beta_{BI} = 10$, $\beta_{IU} = 1$, $p_0 = 30$ dBm, $\sigma^2 = -90$ dBm, and $\epsilon = 0.95$. 
        Each point in the figures is averaged over $1000$ random user locations and channel realizations. 

    \vspace{-5pt}
    \subsection{Algorithm Convergence}\label{subsec:RX-convergence}
        As discussed in Section~\ref{sec:proposed-RX}, the convergence of the proposed algorithms is guaranteed with the alternating optimization process. 
        In particular, the eigenvalue-ratio function and the generalized eigenvalue-ratio function, i.e., $g(\boldsymbol{H})$ for $b\ge 2$ and $g_r(\boldsymbol{H}_r)$ for $b = 1$, are non-decreasing and upper-bounded by $1$. 
        The values of the ratio functions over the iterations for both $b = 1$ and $b\ge 2$ cases are shown in Fig.~\ref{Fig:RX-convergence}. 
        The total number of power measurements is set as $T = 65$. 
        It can be observed that the ratio functions start from small values with the initializations, and converge to $1$ within $5$ iterations. 
        This verifies that the estimated channel autocorrelation matrices obtained by the proposed algorithms approach a rank-one matrix for $b\ge 2$ or a rank-two matrix for $b = 1$, respectively. 

    \vspace{-5pt}
    \subsection{Channel Autocorrelation Matrix Estimation Error}\label{subsec:est-err}
        In this subsection, the estimation error of the proposed channel autocorrelation matrix estimation algorithms is evaluated. 
        For the case of $b\ge 2$, the channel autocorrelation matrix $\bar{\boldsymbol{H}}$ is estimated as $\hat{\boldsymbol{H}}$, and the normalized estimation error is defined as $\mathcal{E}_b = \|\hat{\boldsymbol{H}} - \bar{\boldsymbol{H}}\|_F^2 / \|\bar{\boldsymbol{H}}\|_F^2$, where $\|\cdot\|_F$ denotes the Frobenius norm of a matrix. 
        For the case of $b = 1$, the matrix $\bar{\boldsymbol{H}}_r$ is estimated as $\hat{\boldsymbol{H}}_r$, and thus the normalized estimation error is defined as $\mathcal{E}_b = \|\hat{\boldsymbol{H}}_r - \bar{\boldsymbol{H}}_r\|_F^2 / \|\bar{\boldsymbol{H}}_r\|_F^2$. 
        
        In Fig.~\ref{Fig:RX-error}, ``Trace-min'' represents the solutions for Problems~\eqref{prob:cov-est-find-origin} and~\eqref{prob:cov-est-find-rank-two} for $b\ge 2$ and $b = 1$, respectively, by employing the trace-minimization relaxation method in~\cite{ref:PhaseLift}. 
        It can be observed that the estimation errors of the proposed algorithms decrease rapidly with the number of power measurements and are much smaller than those of the trace-minimization relaxation method. 
        When the number of power measurements is small, there may exist more than one possible solutions for the channel autocorrelation matrix. 
        In this case, even if a rank-one matrix is estimated based on power measurements, it may not be the actual channel autocorrelation matrix $\bar{\boldsymbol{H}}$ or $\bar{\boldsymbol{H}}_r$, which results in a high estimation error. 
        However, when the number of power measurements is sufficiently large, the solution is unique according to Propositions~\ref{prop:solution-set} and~\ref{prop:solution-set-cov-real}, and it is guaranteed that the estimated rank-one matrix approaches the actual channel autocorrelation matrix. 
        
        \begin{figure}[t]
            \begin{center}
                \includegraphics[scale = 0.45]{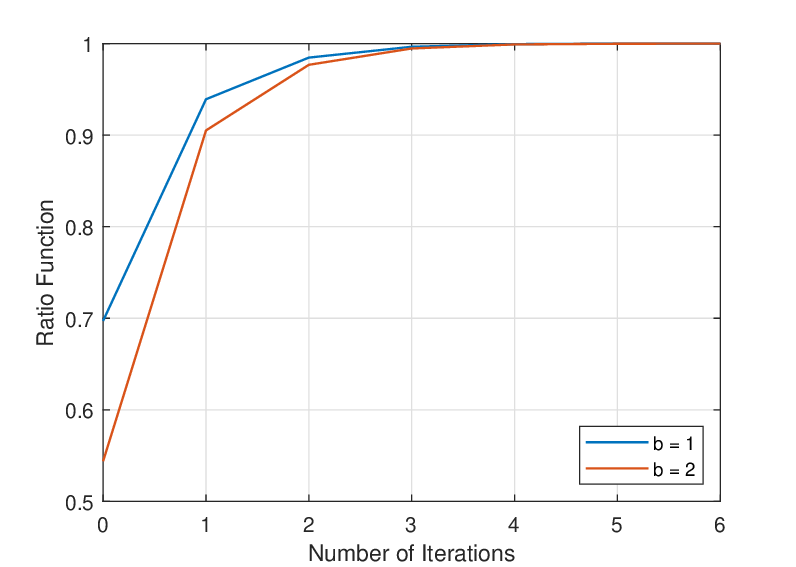}
                \vspace{-7pt}
                \caption{Convergence of the proposed algorithms. }
                \label{Fig:RX-convergence}
            \end{center}
            \vspace{-18pt}
        \end{figure}
        \begin{figure}[t]
            \begin{center}
                \includegraphics[scale = 0.47]{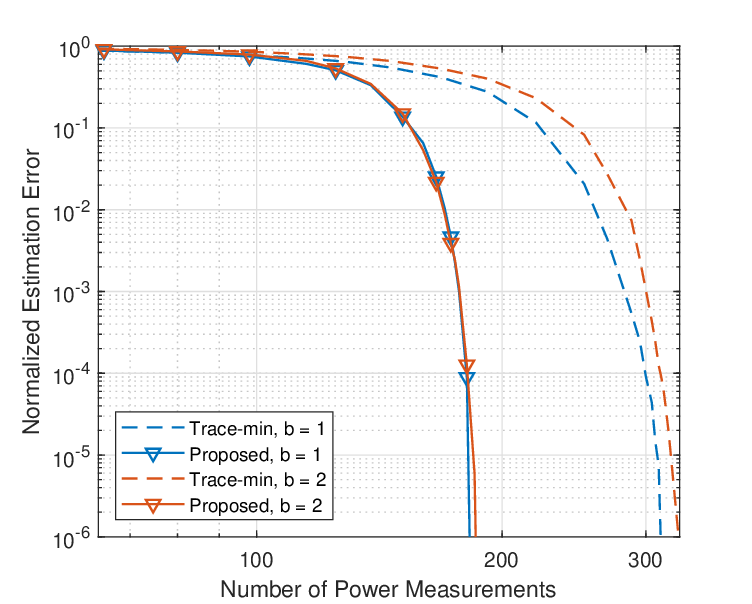}
                \vspace{-7pt}
                \caption{Normalized estimation error of the proposed algorithms. }
                \label{Fig:RX-error}
            \end{center}
            \vspace{-6pt}
        \end{figure}
        
    \subsection{IRS Reflection Design with Estimated Channel}\label{subsec:bf-gain}
        After the channel autocorrelation matrix is estimated, the IRS reflection vector $\boldsymbol{v}$ can be optimized to maximize the effective channel gain, denoted as $\bar{\gamma} = \text{tr}(\bar{\boldsymbol{H}}\boldsymbol{V}) / p_0$, for data transmission. 
        For $b\ge 2$, we apply eigenvalue decomposition to the estimated matrix $\hat{\boldsymbol{H}}$ and define $\hat{\lambda}$ as the largest eigenvalue of $\hat{\boldsymbol{H}}$ and $\hat{\boldsymbol{x}}$ as the corresponding normalized eigenvector. 
        Since $\hat{\boldsymbol{H}}\approx\bar{\boldsymbol{H}}$ is nearly rank-one, the effective channel gain can be approximated as $\bar{\gamma} \approx \text{tr}(\hat{\boldsymbol{H}}\boldsymbol{V}) / p_0 \approx \hat{\lambda}|\hat{\boldsymbol{x}}^H\boldsymbol{v}|^2 / p_0$. 
        Then, the IRS beamforming vector $\boldsymbol{v}$ is optimized to maximize $|\hat{\boldsymbol{x}}^H\boldsymbol{v}|^2$ subject to the discrete phase shift constraint $\boldsymbol{v}\in\Phi_b^N$, which can be solved optimally by using the method proposed in~\cite{ref:optimal-discrete-IRS-vector}. 
        For $b = 1$, similarly, eigenvalue decomposition is applied to the estimated matrix $\hat{\boldsymbol{H}}_r$, where $\hat{\lambda}_1$ and $\hat{\lambda}_2$ denote the first and second largest eigenvalues of $\hat{\boldsymbol{H}}_r$, with $\hat{\boldsymbol{x}}_1$ and $\hat{\boldsymbol{x}}_2$ denoting the corresponding eigenvectors, respectively. 
        As the IRS beamforming vector $\boldsymbol{v}$ is always a real vector for $b = 1$, the effective channel gain can be approximated by 
        \begin{equation}
            \begin{aligned}
                \bar{\gamma} \approx \frac{1}{p_0}\text{tr}(\hat{\boldsymbol{H}}_r\boldsymbol{V}) & \approx \frac{1}{p_0}\left|
                    \hat{\lambda}_1^{\frac{1}{2}}\hat{\boldsymbol{x}}_1^T\boldsymbol{v}
                \right|^2 + \frac{1}{p_0}\left|
                    \hat{\lambda}_2^{\frac{1}{2}}\hat{\boldsymbol{x}}_2^T\boldsymbol{v}
                \right|^2 \\
                & = \frac{1}{p_0}\left|\left(
                    \hat{\lambda}_1^{\frac{1}{2}}\hat{\boldsymbol{x}}_1 + j\hat{\lambda}_2^{\frac{1}{2}}\hat{\boldsymbol{x}}_2
                \right)^H\boldsymbol{v}\right|^2, 
            \end{aligned}
        \end{equation}
        and the optimal vector $\boldsymbol{v}$ can be obtained according to the method proposed in~\cite{ref:optimal-discrete-IRS-vector} {{\color{\highlightcolor}with linear complexity. 
        Therefore, the overall complexity of the IRS reflection design based on the proposed estimation algorithm is still $\mathcal{O}(N^{4.5}I)$. 
        }}

        For comparison, the benchmark schemes for IRS reflection design based on power measurements are listed as follows:
        \textbf{$1)$Upper bound}: 
            The upper bound on the effective channel gain is obtained by optimizing the IRS reflection vector based on the perfect CSI $\boldsymbol{h}$ with the algorithm proposed in~\cite{ref:optimal-discrete-IRS-vector}; 
        \textbf{$2)$RMS} (random-max sampling): 
            A large number of random IRS reflection vectors are applied with $u_n$ uniformly distributed in $\Phi_b$ for $\forall n$ and the one achieving the largest received signal power is used; 
        \textbf{$3)$CSM} (conditional sample mean): 
            This is the method proposed in~\cite{ref:CSM}, where a large number of random IRS reflection vectors are applied, and the empirical expectation of the received signal power is calculated conditioned on $u_n$ fixed at every possible value for $\forall n$. 
            For each element, CSM selects the phase shift that maximizes the empirical expectation conditioned on it. 
        {\color{\highlightcolor}The complexity of RMS and CSM are $\mathcal{O}(N)$ and $\mathcal{O}(NT)$, respectively. }
        
        The effective channel gain obtained based on the estimated channels by the proposed algorithms as well as other benchmark schemes are shown in Fig.~\ref{Fig:RX-gain}. 
        As can be observed, the effective channel gain achieved by the proposed schemes increases rapidly with the number of power measurements, and can approach the upper bound when $T\ge 160$ for both $b = 1$ and $b\ge 2$ cases. 
        As discussed previously, the channel estimation error of the proposed algorithms vanishes quickly when $T$ becomes large, leading to a high effective channel gain. 
        In contrast, the effective channel gains for RMS and CSM increase slowly with $T$, as they do not fully utilize the power measurements for CSI estimation. 
        {\color{\highlightcolor}Thus, although the proposed approach has a higher complexity, it is more efficient than RMS and CSM in improving the effective channel gain because the required number of power measurements to achieve the same performance is much fewer. }

        \begin{figure}[t]
            \begin{center}
                \includegraphics[scale = 0.515]{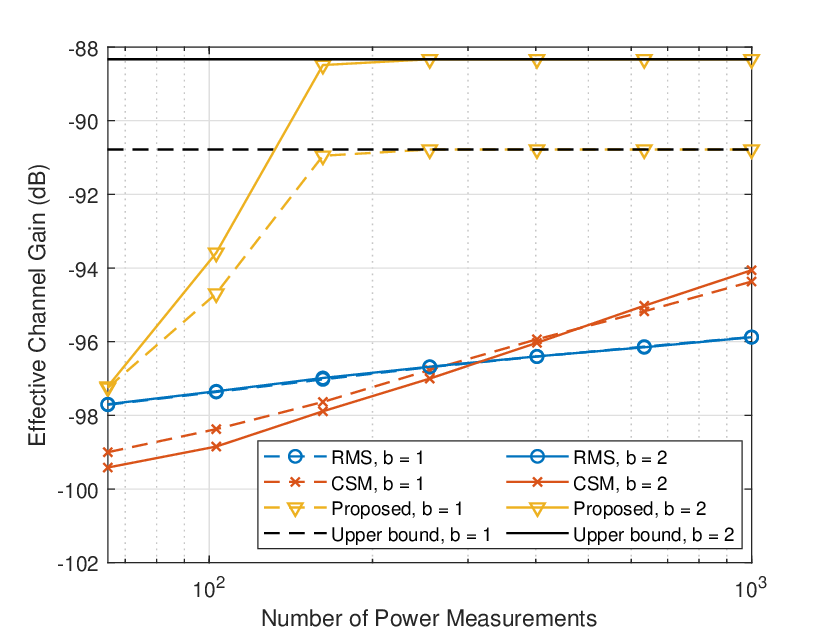}
                \vspace{-7pt}
                \caption{Effective channel gain for different schemes. }
                \label{Fig:RX-gain}
            \end{center}
            \vspace{-10pt}
        \end{figure}

%% file: content/conclusion.tex
\vspace{-3pt}
\section{Conclusion}\label{sec:conclusion}
    This paper proposed a new approach to estimate the autocorrelation matrix of IRS-cascaded channel by leveraging the existing user power measurement mechanism for IRS-assisted communication systems. 
    This approach is practically appealing, as it does not require any change of the channel estimation/training protocol in current wireless systems. 
    It was shown that the IRS channel autocorrelation matrix can be estimated by solving matrix-rank-minimization problems with the alternating optimization method. 
    Simulation results verified the fast convergence and high accuracy of the proposed estimation algorithms and also demonstrated the effectiveness of IRS reflection design based on the estimated channel in improving the effective channel gain for IRS-assisted systems. 
    Although this paper considered the simple setup of single user with slow-fading channels due to the power measurement overhead and relatively high computational complexity, the proposed approaches can be simplified with approximate solutions and thus are extendable to more general setups with multiple users and slow-varying channel statistics (e.g., in an indoor environment), which will be studied in future work. 

\vspace{-3pt}
\section*{Acknowledgement}\label{sec:acknowledgement}
    This work is supported in part by Shenzhen Research Institute of Big Data with the grant No J00120230006,  MOE Singapore under Award T2EP50120-0024, Advanced Research and Technology Innovation Centre of National University of Singapore under Research Grant R-261-518-005-720, and the Guangdong Provincial Key Laboratory of Big Data Computing.